\documentclass[onecolumn,superscriptaddress,11pt,accepted=2017-11-09]{quantumarticle}
\pdfoutput=1
\usepackage[utf8]{inputenc}
\usepackage[english]{babel}
\usepackage[T1]{fontenc}

\usepackage{amsmath}
\usepackage{amsthm}
\usepackage{amssymb}
\usepackage{amsfonts}
\usepackage{graphicx}
\usepackage{color}
\usepackage{hyperref}
\usepackage[algoruled,noline,linesnumbered]{algorithm2e}
\SetNlSty{texttt}{}{}
\SetAlgoNlRelativeSize{0}
\SetNlSkip{2mm}
\SetInd{3mm}{3mm}
\hypersetup{
	colorlinks,
	linkcolor={blue},
	citecolor={blue},
	urlcolor={blue}
}

\usepackage{cleveref}

\newcommand{\rV}{\rVert}
\newcommand{\lV}{\lVert}
\newcommand{\rv}{\rvert}
\newcommand{\lv}{\lvert}
\newcommand{\tV}{\vert\kern-0.25ex\vert\kern-0.25ex\vert}
\newcommand{\la}{\langle}
\newcommand{\ra}{\rangle}
\newcommand{\laa}{\langle\!\langle}
\newcommand{\raa}{\rangle\!\rangle}

\newcommand{\ct}{\ensuremath{^{\dagger}}}

\DeclareMathOperator{\tr}{Tr}
\DeclareMathOperator{\opv}{vec}

\DeclareMathOperator{\spec}{spec}

\newcommand{\cB}{\ensuremath{\mathcal{B}}}
\newcommand{\cC}{\ensuremath{\mathcal{C}}}
\newcommand{\cD}{\ensuremath{\mathcal{D}}}
\newcommand{\cE}{\ensuremath{\mathcal{E}}}

\newcommand{\cG}{\ensuremath{\mathcal{G}}}
\newcommand{\cH}{\ensuremath{\mathcal{H}}}
\newcommand{\cI}{\ensuremath{\mathcal{I}}}

\newcommand{\cL}{\ensuremath{\mathcal{L}}}

\newcommand{\cR}{\ensuremath{\mathcal{R}}}
\newcommand{\cS}{\ensuremath{\mathcal{S}}}

\newcommand{\cU}{\ensuremath{\mathcal{U}}}

\newcommand{\cZ}{\ensuremath{\mathcal{Z}}}

\newcommand{\bbC}{\ensuremath{\mathbb{C}}}

\newcommand{\bbE}{\ensuremath{\mathbb{E}}}

\newcommand{\bbG}{\ensuremath{\mathbb{G}}}

\newcommand{\bbU}{\ensuremath{\mathbb{U}}}

\newcommand{\bbX}{\ensuremath{\mathbb{X}}}

\newcommand{\bbZ}{\ensuremath{\mathbb{Z}}}

\newcommand{\hato}{\ensuremath{\hat{I}}}

\theoremstyle{plain}
\newtheorem{thm}{Theorem}
\newtheorem{lem}[thm]{Lemma}

\theoremstyle{definition}

\theoremstyle{remark}
\newtheorem*{rem}{Remark}

\begin{document}

\title{Randomized benchmarking with gate-dependent noise}
\date{December 18, 2017}
\author{Joel J. Wallman}
\affiliation{Institute for Quantum Computing and Department of Applied 
Mathematics, University of Waterloo, Waterloo, Ontario N2L 3G1, Canada}
\orcid{0000-0001-6943-5334}

\begin{abstract}
We analyze randomized benchmarking for arbitrary gate-dependent noise and prove that the exact impact of gate-dependent noise can be described by a single perturbation term that decays exponentially with the sequence length. 
That is, the exact behavior of randomized benchmarking under general gate-dependent noise converges exponentially to a true exponential decay of exactly the same form as that predicted by previous analysis for gate-independent noise. 
Moreover, we show that the operational meaning of the decay parameter for gate-dependent noise is essentially unchanged, that is, we show that it quantifies the average fidelity of the noise between ideal gates. 
We numerically demonstrate that our analysis is valid for strongly gate-dependent noise models.
We also show why alternative analyses do not provide a rigorous justification for the empirical success of randomized benchmarking with gate-dependent noise.
\end{abstract}

\maketitle

\section{Introduction}

The development of practical, large-scale devices that process quantum 
information relies upon techniques for efficiently characterizing the quality 
of control operations on the quantum level. Fully characterizing quantum 
processes~\cite{Chuang1997,Poyatos1997} requires resources that scale 
exponentially with the number of qubits despite improvements such as compressed 
sensing and direct fidelity 
estimation~\cite{DaSilva2011,Flammia2011,Flammia2012,Reich2013,Kliesch}. 
However, some figures of merit can be efficiently estimated without fully characterizing quantum processes. 
The canonical such figure is the fidelity of an experimental implementation $\tilde{\cG}$ of an 
ideal unitary channel $\cG(\rho) = G\rho G\ct$ averaged over input states
distributed according to the Haar measure $d\psi$,
\begin{align}\label{eq:fidelity}
f(\tilde{\cG},\cG) 
= \int d\psi \tr[\cG(\psi)\tilde{\cG}(\psi)].
\end{align}

At present, the only efficient partial characterizations are 
randomized benchmarking (RB)~\cite{Emerson2005,Levi2007,Knill2008,Dankert2009,Magesan2011} and variants thereof~\cite{Emerson2007,Magesan2012,Wallman2015,Wallman2015b,Wallman2016,Carignan-Dugas2015,Cross2016}. 
The RB protocol of Ref.~\cite{Magesan2011} provides an estimate of the average fidelity over a group of operations from how the fidelity of sequences of operations decays as more operations are applied. 
RB protocols scale efficiently with the number of qubits and are robust to 
state-preparation and measurement (SPAM) errors. Consequently, RB protocols have become an important baseline for the validation and verification of quantum operations~\cite{Corcoles2013,Barends2014} and have recently been used to efficiently optimize over experimental control procedures~\cite{Kelly2014}.

However, rigorous analyses of RB protocols make several key assumptions, namely, that the noise is Markovian and time- and gate-independent, so that RB decay curves provide somewhat heuristic estimates of noise parameters (as do all noise characterization protocols).
Numerical simulations of RB experiments have shown that the standard single-exponential decay curve holds (approximately) for a variety of physically relevant noise models with gate-dependent and non-Markovian noise, although the `true' and estimated fidelity often differ by a factor of approximately two~\cite{Epstein2014,Ball2016}.
However, it has recently been noted that gate-dependent noise fluctuations can dominate the gate-independent decay curve, resulting in a different decay parameter than expected when compared to a standard definition of the average noise~\cite{Proctor}.

In this paper, we prove that, under the assumption of Markovian noise, the full effect of gate-dependent fluctuations can be described by a single perturbation term that decays exponentially with the sequence length $m$. 
We also prove that the RB decay parameter is the fidelity of a suitably defined average noise process $\cE$. 
Informally, our main result (\cref{thm:GD_exact}) is that for gate-dependent, but trace-perserving and Markovian noise, the average survival probability over all RB sequences of length $m$ is
\begin{align}
Ap(\cE)^m + B + \epsilon_m
\end{align}
for some constants $A$ and $B$, where 
\begin{align}\label{eq:prelation}
p(\cE) &= \frac{d f(\cE,\cI)-1}{d-1} \notag\\
\rv\epsilon_m\lv &\leq \delta_1 \delta_2^m,
\end{align}
$\delta_1,\delta_2$ quantify the gate-dependence of the noise and $\delta_2$ is small for good implementations of gates by \cref{thm:smalldelta}. 
We prove all theorems without restricting to trace-preserving noise, as we anticipate that near-term implementations of RB on encoded qubits~\cite{Combes2017} will involve significant loss from the encoded space due to imperfect error-correction procedures.
We also demonstrate that the alternative analyses of Ref.~\cite{Magesan2011,
Magesan2012a, Chasseur2015, Proctor} do not rigorously justify fitting RB decays
to a single exponential in \cref{sec:RBalt}.

This paper is structured as follows. We introduce notation in \cref{sec:notation} and review the RB protocol of Ref.~\cite{Magesan2011} in \cref{sec:RBprotocol}.
We then discuss an ambiguity in the decomposition of noise in \cref{sec:GD_noise} and show how it can be exploited to cancel out the majority of gate-dependent fluctuations in \cref{sec:GD_RB}. 
We illustrate the accuracy of our new analysis numerically in \cref{sec:numerics} and show why alternative analyses of RB for gate-dependent noise do not justify fitting RB experiments to single exponential decays in \cref{sec:RBalt}.
We conclude by discussing some implications of our results for implementing and interpreting RB.

\section{Notation}\label{sec:notation}

We use the following notation throughout this paper.
\begin{itemize}
\item All operators except density operators are denoted by Roman font (e.g., $A$). 
The normalized identity matrix is $\hato_d = I_d/\sqrt{d}$.
\item Ideal channels are denoted by calligraphic font (e.g., $\cC$), where the 
ideal channel $\cU$ corresponding to a unitary operator $U$ is 
$\cU(\rho) = U\rho U\ct$.
\item The unital component of a trace-preserving map is the linear map defined by $\cG_{\mathrm{u}}(A) = \cG(A-\tr(A)/d)$.
\item A noisy implementation of an ideal channel is denoted with an overset
$\tilde{}$ (e.g., $\tilde{\cC}$ denotes the noisy implementation of $\cC$).
\item Channel composition is denoted by (noncommutative) multiplication 
(i.e., $\cC\cB(A) = \cC[\cB(A)]$).
\item Noncommutative products of subscripted objects are denoted by the 
shorthand
\begin{align}
x_{b:a} &= \begin{cases}
x_b x_{b-1}\ldots x_{a+1} x_a & b\geq a, \\
1 & \mbox{otherwise.} \end{cases}
\end{align}
\item Groups and sets will be denoted by blackboard font, e.g., $\bbG$.
\item The uniform average over a set $\bbX$ is denoted by $\bbE_{x\in\bbX} [f(x)] =|X|^{-1}\sum_{x\in \bbX}f(x)$.
\item The trace and operator norms of a matrix are $\lV M \rV_1 = \tr\sqrt{M\ct M}$ and $\lV M \rV = \max_v \lV M v \rV/\lV v\rV$ respectively.
\item The induced trace norm of a linear map $\cL$ is
\begin{align}
\lV \cL \rV_{1\to 1} = \sup_{M:\lV M \rV_1=1} \lV \cL(M) 
\rV_1,
\end{align}
where the maximization is over positive-semidefinite matrices.
\end{itemize}

We use the following matrix representation of quantum channels (variously known as the Liouville, natural and Pauli transfer matrix representation) wherever a concrete representation is required. 
Let $\{e_1,\ldots,e_n\}$ be the canonical orthonormal unit basis of $\bbC^n$ and $\{B_1,\ldots,B_{d^2}\}$ be a trace-orthonormal basis of $\bbC^{d\times d}$, that is, $\tr(B_j\ct B_k) = \delta_{j,k}$. 
Defining the linear map $|*\raa:\bbC^{d\times d}\to\bbC^{d^2}$ by 
\begin{align} 
|A\raa = \sum_j \tr(B_j\ct A)e_j 
\end{align} 
and the adjoint map $\laa A| = |A\raa\ct$, we have $\laa A|B\raa = \tr(A\ct B)$. 
A channel $\cC$ maps $\rho$ to $\cC(\rho)$, which can be represented 
by the matrix 
\begin{align}\label{eq:matrix_rep} 
\cC = \sum_j |\cC(B_j)\raa\!\laa B_j|, 
\end{align} 
where we abuse notation slightly by using the same notation for the abstract channel and its matrix representation. 
For the remainder of the paper, we assume  the operator basis is Hermitian and that $B_1 = \hat{I}_d$, so that the matrix representation of all Hermiticity-preserving maps, including all completely positive and trace-preserving (CPTP) maps, are real-valued. 

\section{Randomized benchmarking protocol}\label{sec:RBprotocol}

The RB protocol for a group of operations $\bbG$ that is also a unitary two-design~\cite{Dankert2009} (e.g., the $n$-qubit Clifford group) is as follows. 
\begin{enumerate}
\item Choose a positive integer $m$.
\item Choose a sequence $\vec{G} = (G_1,\ldots,G_m)\in\bbG^m$ of $m$ elements of $\bbG$ uniformly at random and set $G_{m+1} = (G_{m:1})\ct$.
\item Estimate the expectation value (also known as the survival probability) $Q_{\vec{G}} = \tr[Q\tilde{\cG}_{m+1:1}(\rho)]$ of an observable $Q$ after preparing a $d$-level system in a state $\rho$ and applying the sequence of $m+1$ operations $G_1, G_2,\ldots,G_{m+1}$.
\item Repeat steps 2--3 to estimate the average over random sequences 
	$\bbE_{\vec{G}\in\bbG^m}(Q_{\vec{G}})$ to a desired precision.
\item Repeat steps 1--4 for different values of $m$ and fit to the decay model 
\begin{align}\label{eq:fidelity_decay}
\bbE_{\vec{G}\in\bbG^m}(Q_{\vec{G}}) = Ap^m + B
\end{align}
to estimate $p$~\cite{Magesan2011}.
\end{enumerate}

Despite the ubiquitous usage of RB, best practice choices for the sequence lengths and the number of sequences at each length are still unknown, though see Refs.~\cite{Epstein2014,Wallman2014,Granade2014} for some guidance.

A unitary two-design is any set of unitary channels $\bbG$ such that uniformly sampling $\bbG$ reproduces the first and second moments of the Haar measure.
The canonical example of a unitary two-design is the $n$-qubit Clifford group.
The unitary two-design condition can also be stated in terms of how a general channel is `twirled' by the group, where the  twirl of a channel $\cC$ over a group $\bbG$ is $\bbE_{G\in\bbG} (\cG\ct \cC\cG)$.
A unitary two-design is any group $\bbG$ such that any CPTP channel is `twirled' into a depolarizing channel~\cite{Dankert2009,Nielsen2002}
\begin{align}
\cD_p(\rho) = p\rho + \frac{1-p}{d}I_d.
\end{align}
Unitary two-designs also twirl  a general linear map $\cC$ into a depolarizing 
map composed with a channel that uniformly decreases the trace, 
that is, into a channel of the form
\begin{align}
\cD_{p,t}(\rho) = \frac{t}{d}I_d + p(\rho-\frac{1}{d}I_d),
\end{align}
where the values of $p$ and $t$ can be computed from the following lemma~\cite{Wallman2014}.
Note that $p$ and $t$ are functions from linear maps to the real numbers, however, we will often omit the argument when it is clear from the context.

\begin{lem}\label{lem:twirl}
For any unitary two-design $\bbG$ and channel $\cC$,
\begin{align}
\bbE_{G\in\bbG} (\cG\ct \cC\cG) = \cD_{p,t},
\end{align}
where
\begin{align}\label{eq:decay_parameter}
t(\cC) &= \frac{\tr[\cC(I_d)]}{d} \notag\\
p(\cC) &= \frac{\tr(\cC) - t}{d^2-1} 
= \frac{df(\cC) - t}{d-1}.
\end{align}
\end{lem}

In \cref{thm:GD_exact}, we prove that \cref{eq:fidelity_decay} holds, up to a
small and exponentially decaying perturbation, for gate-dependent noise that is time-independent and Markovian, generalizing the original analysis to more realistic noise. 

\section{Representing gate-dependent noise}\label{sec:GD_noise}

Let $\tilde{\cG}$ be an experimental implementation of an ideal unitary channel $\cG(\rho) = G\rho G\ct$ with Markovian noise.
It is convenient to decompose $\tilde{\cG}$ into an ideal gate and a fictitious noise process, for example, $\tilde{\cG} = \cG\cE$.
However, we can also write $\tilde{\cG} = \cL_G \cG \cR_G$ for any suitable $\cL_G$ and $\cR_G$.
We now identify an approximate decomposition that will greatly simplify the analysis of RB with gate-dependent noise.

\begin{thm}\label{thm:avnoise}
	Let $\bbG$ be a unitary two-design, $\{\tilde{\cG}:G\in\bbG\}$ be a corresponding set of Hermiticity-preserving maps, and $p$ and $t$ be the largest eigenvalues of $\bbE_{G\in\bbG} (\cG_{\mathrm{u}} \otimes \tilde{\cG})$ and $\bbE_{G\in\bbG} (\tilde{\cG})$ in absolute value respectively, where $\cG_{\mathrm{u}}(A) = \cG(A-\tr(A)/d)$.
There exist linear maps $\cL$ and $\cR$ such that
\begin{subequations}\label{eq:conditions}
\begin{align}
\bbE_{G\in\bbG}(\tilde{\cG}\cL\cG\ct) &= \cL \cD_{p,t} \label{eq:lcona}\\
\bbE_{G\in\bbG}(\cG\ct \cR\tilde{\cG}) &= \cD_{p,t}\cR \label{eq:rcona}\\
\bbE_{G\in\bbG}(\cG\cR\cL\cG\ct) &= \cD_{p,t} \label{eq:LRscale}.
\end{align}
\end{subequations}
\end{thm}

\begin{rem}
When the noise is independent of the gate, that is, $\tilde{\cG} = \cL\cG\cR$ for all $G\in\bbG$ for some $\cL$ and $\cR$, then $\cL$ and $\cR$ are solutions to \cref{eq:conditions}, as can be seen by substituting $\tilde{\cG} = \cL\cG\cR$ into \cref{eq:conditions} and using \cref{lem:twirl}.

For gate-dependent noise such that $\cR$ is invertible, we can set $\tilde{\cG} = \cL_G\cG\cR$, so that the noise between two ideal gates $\cG$ and $\cH$ is $\cR\cL_G$. 
Substituting $\tilde{\cG} = \cL_G\cG\cR$ into \cref{eq:rcona} and multiplying by $\cR^{-1}$ gives
\begin{align}
\bbE_{G\in\bbG}(\cG\ct \cR \cL_G \cG) = \cD_{p,t}.
\end{align}
As the Haar measure is unitarily invariant and $p$ and $t$ are convex functions, 
\begin{align}\label{eq:avGDnoise}
p = p(\bbE_{G\in\bbG} [\cG\ct\cR\cL_G\cG]) = \bbE_{G\in\bbG} (p[\cR\cL_G]) \notag\\
t = t(\bbE_{G\in\bbG} [\cG\ct\cR\cL_G\cG]) = \bbE_{G\in\bbG} (t[\cR\cL_G])
\end{align}
give the fidelity and loss of trace of the average noise between ideal gates to the identity via \cref{lem:twirl}.
If $\cR$ is not invertible, we can introduce an arbitrarily small perturbation into the $\tilde{\cG}$ to make $\cR$ invertible and so \cref{eq:avGDnoise} will hold to an arbitrary precision.

When $\tilde{\cG}= \cG\cD_{p,t}$, $\bbE_{G\in\bbG} (\cG_{\mathrm{u}} \otimes \tilde{\cG})$ and $\bbE_{G\in\bbG} (\tilde{\cG})$ each have only one nonzero eigenvalue, namely, $p$ and $t$ respectively~\cite{Wallman2015,Wallman2015b}.
By the Bauer-Fike theorem~\cite{Bauer1960}, all the eigenvalues $\eta$ of a matrix $M+\delta M$ satisfy 
\begin{align} 
\min_{\lambda\in\spec(M)} \lv \eta -\lambda\rv \leq \lV\delta M\rV, 
\end{align} 
where $\spec(M)$ is the spectrum of $M$.  
As the operator norm is submultiplicative, obeys the triangle inequality, and $\lV A\otimes B\rV = \lV A\rV \lV B\rV$, the largest eigenvalues of $\bbE_{G\in\bbG} (\cG_{\mathrm{u}} \otimes \tilde{\cG})$ and $\bbE_{G\in\bbG} (\tilde{\cG})$ will differ from $p'$ and $t'$ by at most $\bbE_{G\in\bbG} \lV \tilde{\cG} - \cG\cD_{p',t'}\rV$ for any scalars $p'$ and $t'$.
Therefore the largest eigenvalues in absolute value will generally be unique and hence must be real as they are the eigenvalues of real matrices.
\end{rem}

\begin{proof}
\Cref{eq:lcona,eq:rcona} are essentially a pair of eigenvalue equations, except that $\bbG$ acts irreducibly on $|\hato_d\raa$ and its orthogonal complement.
To utilize this structure, let $\cL = |L\raa\!\laa \hato_d| + \cL'$ and
$\cR = |\hato_d\raa\!\laa R| + \cR'$ with $\cL'|\hato_d\raa=0$ and 
$\laa \hato_d|\cR' = 0$ without loss of generality.
With these definitions and as $\cG|\hato_d\raa = |\hato_d\raa$ for any unitary (or unital) channel, we can separate each equation in \cref{eq:lcona,eq:rcona} into two independent equations 
\begin{subequations}\label{eq:mncon}
\begin{align}
\bbE_{G\in\bbG} (\tilde{\cG})|L\raa &= t|L\raa \label{eq:m2con}\\
\bbE_{G\in\bbG} (\tilde{\cG} \cL' \cG_{\mathrm{u}}\ct) &= p\cL' \label{eq:n2con}\\
\laa R|\bbE_{G\in\bbG} (\tilde{\cG}) &= t\laa R| \label{eq:m1con}\\
\bbE_{G\in\bbG} (\cG_{\mathrm{u}}\ct \cR' \tilde{\cG}) &= p\cR' \label{eq:n1con},
\end{align}
\end{subequations}
where restricting to $\cG_{\mathrm{u}}$ enforces the orthogonality constraints on
$\cL'|\hato_d\raa=0$ and $\laa \hato_d|\cR'=0$. 
\Cref{eq:m1con,eq:m2con} are left- and right-eigenvector problems of a real matrix and so $t$ can be set to any eigenvalue of $\bbE_{G\in\bbG} (\tilde{\cG})$ and $L$ and $R$ to the corresponding non-trivial eigenvectors (recall that the set of left- and right-eigenvalues of a real matrix are identical).
We choose the largest eigenvalue as it will result in the smallest $\Delta_G = \tilde{\cG} - \cL\cG\cR$.

To turn \cref{eq:n1con,eq:n2con} into eigenvalue equations, we can use the
vectorization map that maps a matrix to a vector by stacking the columns
vertically, that is,
\begin{align}
\opv(\sum_j v_j e_j\ct) = \sum_j e_j\otimes v_j.
\end{align}
The vectorization map satisfies the identity
\begin{align}\label{eq:vec_product}
\opv (ABC) = (C^T\otimes A)\opv (B),
\end{align}
Applying \cref{eq:vec_product} to \cref{eq:n1con,eq:n2con} and using $\cG^T = 
\cG\ct$ (as the matrix basis is Hermitian) gives
\begin{align}
\bbE_{G\in\bbG} (\cG_{\mathrm{u}} \otimes \tilde{\cG})\opv(\cL') &= p \opv(\cL') \notag\\
\bbE_{G\in\bbG} (\tilde{\cG}\otimes \cG_{\mathrm{u}})^T\opv(\cR') &= p \opv(\cR'),\label{eq:primeconditions}
\end{align}
and so $p$ can be set to any eigenvalue of $\bbE_{G\in\bbG} (\cG_{\mathrm{u}} \otimes \tilde{\cG})$ and $\cL'$ and $\cR'$ to the corresponding non-trivial eigenvectors (or equivalently, to any eigenvalue of $\bbE_{G\in\bbG} (\tilde{\cG}\otimes \cG_{\mathrm{u}})$ as the two are related by a unitary transformation and so have the same spectrum).

The maps $\cL$ and $\cR$ identified are solutions to eigenvector equations and so are only determined up to a normalization constant. 
As $\cD_{q,u}$ commutes with $\cG$ for all scalars $q$ and $v$ and all $\cG\in\bbG$, we can set $\cL\to \cL D_{q,u}$ and $\cR\to \cD_{r,v}\cR$ to satisfy \cref{eq:LRscale}.
\end{proof} 

While applying a gauge transformation $\tilde{\cG}\to\cS^{-1}\tilde{\cG}\cS$ changes $\cL$ and $\cR$ to $\cL\to \cS\cL$ and $\cR\to\cR\cS^{-1}$ respectively, the noise between ideal gates, namely, $\cR\cL$, is gauge invariant and so the present analysis is gauge-invariant unlike the original analysis of RB as shown in Ref.~\cite{Proctor}.

We now prove that the solutions $\cL$ and $\cR$ to \cref{eq:conditions} can also be chosen to ensure that $\Delta_G = \tilde{\cG} - \cL\cG\cR$ is small when $\tilde{\cG}\approx\cG\cD_{p,t}$. To prove this, it is convenient to transform to the gauge $\tilde{\cG}^{\cI}$ such that $\cL=\cI$ (i.e., set $\cS = \cL^{-1}$).
In this gauge, \cref{eq:LRscale} becomes 
\begin{align}\label{eq:Rscale}
\bbE_{G\in\bbG}(\cG\cR\cG\ct) = \cD_{p,t}.
\end{align}
In particular, $\Delta_G$ is on the order of $\lV\tilde{\cG}-\cG\cD_{p,t}\rV\in O(\sqrt{1-p})$.

\begin{thm}\label{thm:smalldelta}
Let $\bbG$ be a unitary two-design, $\{\tilde{\cG}:G\in\bbG\}$ be a corresponding set of Hermiticity-preserving maps, and $p$ and $t$ be the largest eigenvalues of $\bbE_{G\in\bbG} (\cG_{\mathrm{u}} \otimes \tilde{\cG})$ and $\bbE_{G\in\bbG} (\tilde{\cG})$ in absolute value respectively, where $\cG_{\mathrm{u}}(A) = \cG(A-\tr(A)/d)$.
There exists a linear map $\cR$ satisfying \cref{eq:conditions} with $\cL=\cI$
such that
\begin{align}
\lV\tilde{\cG}^{\cI} - \cG\cR\rV \leq 
\lV\tilde{\cG}^{\cI} - \cG\cD_{p,t}\rV + \frac{\lV \bbE_{G\in\bbG} (\cG\ct\tilde{\cG}^{\cI}) - \cD_{p,t}\rV }{1 - \bbE_{G\in\bbG} (\lV \tilde{\cG}^{\cI} - \cG\cD_{p,t}\rV)\lV\cD_{1/p,1/t}\rV}.
\end{align}
\end{thm}

\begin{proof}
By the triangle inequality, submultiplicativity, and unitary invariance of the operator norm,
\begin{align}\label{eq:bound_terms} 
\lV \tilde{\cG}^{\cI} - \cG\cR\rV
&\leq \lV \tilde{\cG}^{\cI} - \cG\cD_{p,t}\rV + \lV \cG(\cD_{p,t}-\cR) \rV \notag\\
&\leq \lV \tilde{\cG}^{\cI} - \cG\cD_{p,t}\rV + \lV \cD_{p,t}-\cR \rV
\end{align} 
for all $G\in\bbG$.

\Cref{eq:bound_terms} holds for any $\cR$. 
Now let $\cR$ be a solution to \cref{eq:conditions} for $\tilde{\cG}^{\cI}$ and expand $\cR = \cD_{p,t} + \cR_1$ and $\tilde{\cG}^{\cI} = \cG\cD_{p,t} + \tilde{\cG}_1$. 
Substituting these expansions into \cref{eq:rcona} and using \cref{eq:Rscale} gives
\begin{align}\label{eq:pert_con}
\cD_{p,t}^2 + \cD_{p,t}\cR_1
&= \bbE_{G\in\bbG}(\cG\ct\cR\cG\cD_{p,t}) + \bbE_{G\in\bbG}(\cG\ct\cR\tilde{\cG}_1) \notag\\
&= \cD_{p,t}^2 + \bbE_{G\in\bbG}(\cG\ct\cR\tilde{\cG}_1)
\end{align}
Canceling the common term in \cref{eq:pert_con}, multiplying both sides by $\cD_{1/p,1/t}$ from the left, taking the operator norm and using the triangle inequality and submultiplicativity of the operator norm gives
\begin{align}
\lV \cR_1 \rV 
&\leq 
\lV \bbE_{G\in\bbG} (\cG\ct\tilde{\cG}_1)\rV 
+ \lV \bbE_{G\in\bbG} (\cD_{1/p,1/t} \cG\ct \cR_1\tilde{\cG}_1) \rV \notag\\
&\leq \lV \bbE_{G\in\bbG} (\cG\ct\tilde{\cG}_1)\rV + \bbE_{G\in\bbG} (\lV \tilde{\cG}_1\rV \lV\cR_1\rV \lV \cD_{1/p,1/t}\cG\ct\rV).
\end{align}
Rearranging and using the unitary invariance of the operator norm gives
\begin{align}
\lV \cR_1\rV &= \lV \cD_{p,t}-\cR \rV \notag\\
&\leq \frac{\lV \bbE_{G\in\bbG} (\cG\ct\tilde{\cG}_1)\rV }{1 - \bbE_{G\in\bbG} (\lV \tilde{\cG}_1\rV) \lV\cD_{1/p,1/t}\rV} \notag\\
&=\frac{\lV \bbE_{G\in\bbG} (\cG\ct\tilde{\cG}^{\cI}) - \cD_{p,t}\rV }{1 - \bbE_{G\in\bbG} (\lV \tilde{\cG}^{\cI} - \cG\cD_{p,t}\rV) \lV\cD_{1/p,1/t}\rV}.
\end{align}
\end{proof}

\section{Analyzing RB with arbitrarily gate-dependent noise}\label{sec:GD_RB}

Experimental implementations of RB almost always produce relatively clean exponential decays. 
We now explain this empirical success for general gate-dependent noise by proving that the average survival probability over RB sequences of length $m$ is equal to the model of Ref.~\cite{Magesan2012a} for gate-independent noise with a single perturbation term for a suitably defined average noise model. 
Moreover, the decay constants correspond to the average loss of trace and fidelity of the noise between ideal gates via \cref{eq:avGDnoise}.
The perturbation term will generally be negligible for noise that is close to the identity by \cref{thm:smalldelta} and for $m\geq 3$ (for numerical values for a single qubit, see \cref{fig:GDnumerics}).
Even under this condition, more significant (and potentially oscillating) perturbations are possible for $m\approx 1$, potentially explaining experimentally observed fluctuations as seen in, e.g., \cite[Fig.~3b]{Corcoles2013}.
While we have not investigated how the perturbation terms scale with dimension, note that we only need $\delta_2\lessapprox 1/2$ in order for the decay to be essentially a single exponential for $m\geq 10$, etc, as $\delta_2$ is exponentially suppressed.

\begin{thm}\label{thm:GD_exact}
For any unitary two-design $\bbG$ and set of Hermiticitiy preserving linear maps $\{\tilde{\cG}:G\in\bbG\}$, the average survival probability over all randomized benchmarking sequences of length $m$ is
\begin{align}\label{eq:exact_decay}
\bbE_{\vec{G}\in\bbG^m} (Q_{\vec{G}})
&= A p^m + B t^m + \epsilon_m
\end{align}
for some constants $A$ and $B$, where $p$ and $t$ give the fidelity and loss of trace of the average noise between ideal gates via \cref{eq:avGDnoise}. 
Moreover, the perturbation term $\epsilon_m$ satisfies
\begin{align}
\lv \epsilon_m \rv\leq \delta_1  \delta_2^m
\end{align}
for some $\delta_1$ and $\delta_2$ that quantify the amount of gate dependence.
\end{thm}

\begin{rem}
\Cref{eq:exact_decay} reduces to the standard single-exponential model plus an exponentially decreasing perturbation when $\tilde{\cG}$ is a CPTP map for all $G$ as then $\bbE_{G\in\bbG} (\tilde{\cG}\otimes \cG)$ and $\bbE_{G\in\bbG} (\cG \otimes \tilde{\cG})$ are both CPTP maps and so all eigenvalues are in the unit disc and the eigenvalue with eigenvector close to $|\hato_d\raa$ is 1~\cite{Perez-Garcia2006}.

The constants $A$ and $B$ and the perturbation term are
\begin{align}
A &= \laa Q|\cL|\cR(\rho)-I_d/d\raa \notag\\
B &= \laa Q|\cL|I_d/d\raa \notag\\
\epsilon_m &= \laa Q|\bbE_{\vec{G}\in\bbG^m}(\Delta_{m+1:1})|\rho\raa.
\end{align}
with $\Delta_j = \Delta_{G_j}= \tilde{\cG}_j - \cL\cG_j \cR$ and where $\cL$ and $\cR$ are solutions to \cref{eq:conditions}. 
The quantities bounding the perturbation term are
\begin{align}\label{eq:deltas}
\delta_1 &= \max_G \lV\Delta_G \rV_{1\to 1}\lV\rho\rV_1 \lV Q\rV_1\notag\\
\delta_2 &= \bbE_{G\in\bbG}(\lV\Delta_G\rV).
\end{align}
which are in turn bounded in \cref{thm:smalldelta} in the gauge where $\cL=\cI$.
Note that $\lV\rho\rV_1,\lV Q\rV_1\leq 1$ in a standard gauge, however, the gauge transformation required to set $\cL = \cI$ may change these values and introduce negative eigenvalues to $\rho$ so that the maximization for the induced norm in $\delta_1$ is not restricted to positive semi-definite inputs.
\end{rem}

\begin{proof}
The average map over all randomized benchmarking sequences $\vec{G}$ of length
$m$ is 
\begin{align}
\bbE_{\vec{G}\in\bbG^m}(\tilde{\cG}_{m+1:1}).
\end{align}
Let $\cL$ and $\cR$ satisfy
\begin{subequations}\label{eq:mapcon}
\begin{align}
\bbE_{G\in\bbG} (\tilde{\cG}\cL\cG\ct) &= \cL\cD_{p,t} \label{eq:lcon}\\
\bbE_{G\in\bbG} (\cG\ct \cR \tilde{\cG}) &= \cD_{p,t}\cR \label{eq:rcon}\\
\bbE_{G\in\bbG}(\cG\cR\cL\cG\ct) &= \cD_{p,t}
\end{align}
\end{subequations}
and define
\begin{align}
\Delta_j = \Delta_{G_j} = \tilde{\cG}_j - \cL\cG_j \cR.
\end{align}
Then for any integer $1\leq j \leq m+1$ we can write
\begin{align}\label{eq:splitting}
\bbE_{\vec{G}\in\bbG^m}(\tilde{\cG}_{m+1:j}\Delta_{j-1:1})
&= \bbE_{\vec{G}\in\bbG^m}(\tilde{\cG}_{m+1:j+1} \cL\cG_j \cR \Delta_{j-1:1}) 
+ \bbE_{\vec{G}\in\bbG^m}(\tilde{\cG}_{m+1:j+1}\Delta_j\Delta_{j-1:1}) \notag\\
&= \bbE_{\vec{G}\in\bbG^m}(\tilde{\cG}_{m+1:j+1} \cL\cG_j \cR \Delta_{j-1:1}) 
+ \bbE_{\vec{G}\in\bbG^m}(\tilde{\cG}_{m+1:j+1}\Delta_{j:1}),
\end{align}
where $\Delta_{0:1} = I_d = \tilde{\cG}_{m+1:m+2}$ by convention.
Moreover, as $G_{m+1:1} = I_d$ and any set of $m$ elements of $\vec{G}$ are 
statistically independent and uniformly distributed, we can apply 
\cref{eq:lcon} recursively and use the fact that $\cD_{p,t}$ commutes with 
unital channels to obtain
\begin{align}\label{eq:j1term}
\bbE_{\vec{G}\in\bbG^m}(\tilde{\cG}_{m+1:2} \cL\cG_1 \cR)
&= \bbE_{\vec{G}\in\bbG^m}(\tilde{\cG}_{m+1:2} \cL [\cG_{m+1:2}]\ct \cR) \notag\\
&= \cL\cD_{p,t}^m \cR
\end{align}
for $j=1$. Similarly, we can use \cref{lem:twirl,eq:mapcon} for $j\geq 2$ to obtain
\begin{align}\label{eq:otherj}
\bbE_{\vec{G}\in\bbG^m}[\tilde{\cG}_{m+1:j+1} \cL\cG_j \cR \Delta_{j-1:1}]
&= \bbE_{\vec{G}\in\bbG^m}[\tilde{\cG}_{m+1:j+1} \cL (\cG_{m+1:j+1})\ct 
(\cG_{j-2:1})\ct \cG_{j-1}\ct \cR \Delta_{j-1}\Delta_{j-2:1}] \notag\\
&= \cL\cD_{p,t}^{m-j}\bbE_{\vec{G}\in\bbG^m}[(\cG_{j-2:1})\ct \cG_{j-1}\ct \cR(\tilde{\cG}_{j-1}-\cL\cG_{j-1}\cR) \Delta_{j-2:1}] \notag\\
&= \cL\cD_{p,t}^{m-j} \bbE_{\vec{G}\in\bbG^m}[(\cG_{j-2:1})\ct 
\bbE_{G_{j-1}}(\cG_{j-1}\ct \cR\tilde{\cG}_{j-1}- \cD_{p,t}\cR) \Delta_{j-2:1}] \notag\\
&= 0.
\end{align}
Therefore we can apply \cref{eq:splitting,eq:j1term} and then recursively apply
\cref{eq:splitting,eq:otherj} to obtain
\begin{align}
\bbE_{\vec{G}\in\bbG^m} \tilde{\cG}_{m+1:1} 
&= \cL\cD_{p,t}^m \cR 
+ \bbE_{\vec{G}\in\bbG^m}[\tilde{\cG}_{m+1:2}\Delta_{1}] \notag\\
&= \cL\cD_{p,t}^m \cR + \bbE_{\vec{G}\in\bbG^m}[\Delta_{m+1:1}].
\end{align}

The average survival probability over all random sequences is then
\begin{align}
\bbE_{\vec{G}\in\bbG^m} Q_{\vec{G}} 
&= \laa Q|\cL\cD_{p,t}^m \cR|\rho\raa + \epsilon_m \notag\\
&= t^m \laa Q|\cL|I_d/d\raa + p^m \laa Q|\cL|\cR(\rho)-I_d/d\raa + \epsilon_m
\end{align}
where
\begin{align}
\epsilon_m &= \laa Q|\bbE_{\vec{G}\in\bbG^m}(\Delta_{m+1:1})|\rho\raa.
\end{align}

We now bound the perturbation term. By the triangle inequality and submultiplicativity,
\begin{align}
\lV \bbE_{\vec{G}\in\bbG^m} (\tilde{\cG}_{m+1:1}) - \cL\cD_{p,t}^m \cR \rV_{1\to1}
&= \lV \bbE_{\vec{G}\in\bbG^m}(\Delta_{m+1:1}) \rV_{1\to1} \notag\\
&\leq \bbE_{\vec{G}\in\bbG^m}\lV \Delta_{m+1:1} \rV_{1\to1} \notag\\
&\leq \bbE_{\vec{G}\in\bbG^m}\lV\Delta_1\rV_{1\to1} \prod_{j= m+1}^2 \lV \Delta_j\rV \notag\\
&\leq \max_G \lV\Delta_G \rV_{1\to 1} \bbE_{\vec{G}\in\bbG^m}\prod_{j= m+1}^2  \lV\Delta_j\rV \notag\\
&\leq \max_G \lV\Delta_G \rV_{1\to 1} [\bbE_H \lV\Delta_H\rV]^m .
\end{align}
\end{proof}

\section{Numerics}\label{sec:numerics}

\begin{figure}
	\includegraphics[width=\linewidth]{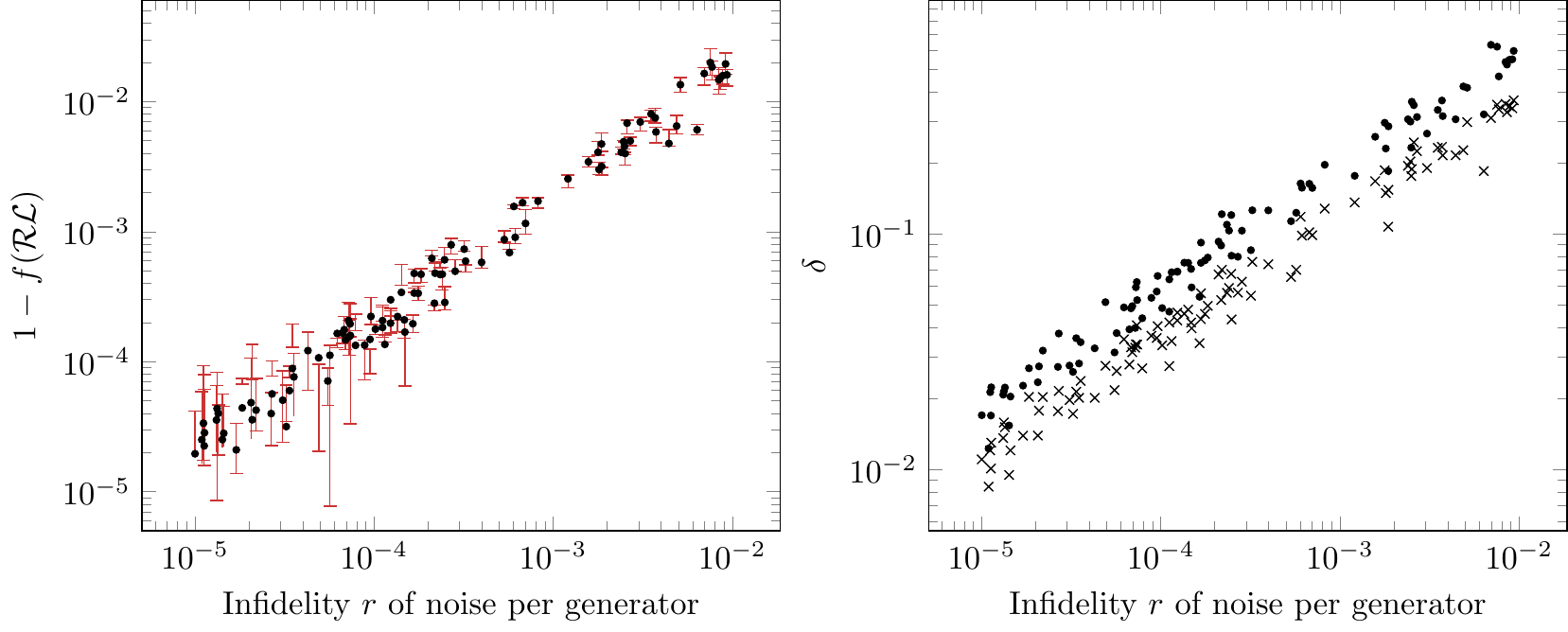}
	\caption{
Left: upper- and lower- edges of 90\% confidence intervals for the decay
parameter from the numerical simulations of RB experiments described in
\cref{sec:numerics} fitted neglecting the perturbation term $\delta_1\delta_2^m$
(red intervals) and the corresponding values of $p$ from \cref{thm:GD_exact}
(black dots).
Right: numerical values of $\delta_1$ (circles) and $\delta_2$ (crosses) from \cref{thm:GD_exact}. 
For all simulated experiments, the perturbation term scales as $\delta_1 \delta_2^m \in O(r^{-m/2})$ and so is negligible for $m\geq 3$.}
	\label{fig:GDnumerics}
\end{figure}

In \cref{fig:GDnumerics} we illustrate the reliability of our 
analysis by comparing the estimates obtained by fitting the simulated RB 
experiments described below under independent random unitary noise on a set of
generating gates to the decay parameters predicted by \cref{thm:GD_exact} and
by plotting the size of the perturbation terms $\delta_1$ and $\delta_2$
calculated in the gauge where $\cL = \cI$.
Mathematica code is available at
\url{https://github.com/jjwallman/numerics}.

For simplicity, we simulate RB experiments using the (projective) group
\begin{align}
\bbG = \{T^t P:t\in\bbZ_3,P\in\{I,X,Y,Z\}\}
\end{align}
where $I,X,Y,Z$ are the standard single-qubit Pauli matrices and
\begin{align}
T = \frac{1}{\sqrt{2}}\begin{pmatrix} 1 & -i \\ 1 & i \end{pmatrix},
\end{align}
which is a subgroup of the single-qubit Clifford group that is also a unitary two-design.
For our simulations, the noisy implementation of $G = T^t P\in\bbG$ is $\tilde{\cG} = \cG'$ where
\begin{align}
G' = U T^t V P
\end{align}
where $U$ and $V$ are independently sampled from the set $\bbU_r\subset
\mathrm{SU}(d)$ satisfying $1 - f(\cU,\cI) = r$, that is, the set of unitary 
channels of fixed infidelity $r$. We then calculate $p(\cR\cL)$ from
\cref{thm:avnoise} and $\delta_1,\delta_2$ from \cref{eq:deltas}.
We set the state to $|0\ra$ and the measurement to be a computational basis
measurement and do not include SPAM errors for simplicity.

The set $\bbU_r$ can be sampled by sampling $V\in\mathrm{SU}(d)$ uniformly from the Haar measure and setting
\begin{align}
U = V\exp(-i\theta Z) V\ct = \cos\theta I - i \sin\theta VZV\ct,
\end{align}
which gives the unique unitarily invariant distribution over $\bbU_r$ for some $r$ as a function of $\theta$.
In order for the infidelity of the resulting gates to be $r$, we require~\cite{Nielsen2002}
\begin{align}
r = \frac{4-\lv \tr U\rv^2 }{6} = \frac{2 - \cos^2\theta}{3},
\end{align}
or $\theta = \arcsin(\sqrt{3r/2})$.

A single RB experiment for a fixed value of $m$ then consisted of the following. 
\begin{enumerate}
	\item Randomly choose $\vec{a}\in \{0,1,2\}^m$ and $\vec{b}\in\{1,2,3,4\}^m$.
	\item Set $G_j = T^{a_j} P_{b_j} P_{b_{j-1}} T^{-a_{j-1}} = T^{a_j - a_{j-1}} P_{\tau_j}$ where $(P_0,P_1,P_2,P_3) = (I,X,Y,Z)$, $a_{m+1} = a_0$, $ b_{m+1} = b_0 = 1$, and $P_{\tau_j} = T^{a_{j-1}} P_{b_j} P_{b_{j-1}} T^{-a_{j-1}}$.
	\item Calculate $\lv \la 0\rv G'_{m+1:1} \lv 0 \ra \rv^2$
\end{enumerate}
We repeat the above steps for 100 random sequences for each value of $m\in \{4,8,16,...,2048\}$ and fit the averages for each sequence to the model
$\bbE_{\vec{\cC}} = 0.5p_{est}^m + 0.5$ using Mathematica's NonlinearModelFit
function weighted by the variance over the 100 sequences at each $m$, where we have set $A=B=0.5$ because of the choice of SPAM and because the noise is unital.

\section{Alternative analysis of RB with gate-dependent noise}\label{sec:RBalt}

We now show how the analyses of gate-dependent noise of Refs.~\cite{Magesan2011,
Magesan2012a, Chasseur2015, Proctor} are too loose to justify fitting RB decay
curves to a single exponential in practical regimes. We will first show that
the higher-order terms in the original analysis of RB give a systematic
uncertainty of the infidelity that dominates the infidelity, even when
the fidelity is dominated by stochastic errors.
We will then discuss subsequent analysis and how it also results in large
systematic errors.

\subsection{Initial RB Analysis}

The initial method of analyzing RB under gate-dependent noise was to perform a
perturbation expansion to perturb about a perfect twirl~\cite{Magesan2011,
Magesan2012a} of the average noise,
\begin{align}\label{eq:av_noise}
\cE = \bbE_{G\in\bbG}(\cG\ct \tilde{\cG}).
\end{align}
Substituting $\tilde{\cG} = \cG\cE + \cG\Delta_G$ into the average map over all 
sequences,
\begin{align}
\bbE_{\vec{G}\in\bbG^m}(\tilde{\cG}_{m+1:1})
\end{align}
and using the triangle inequality and submultiplicativity, the $k$th order terms
in $\Delta_G$ contribute at most
\begin{align}\label{eq:perturbation_bound}
\binom{m+1}{k}\gamma^k = \binom{m+1}{k}\bbE_{G\in\bbG} (\lV \Delta_G \rV_{1\to 1})
\end{align}
to the average survival probability over all sequences of length $m$, $\laa
Q\rv \bbE_{\vec{G}\in\bbG^m}(\tilde{\cG}_{m+1:1}) \lv\rho\raa$.
Ref.~\cite{Magesan2012a} also obtained similar conditions for time-dependent
noise and derived the first-order correction.
As stated in Ref.~\cite{Magesan2011}, higher-order terms will be negligible
provided $m\gamma \ll 1$, so that there is a tension between the amount of
gate-dependent noise and the values of $m$ required to obtain a reasonable fit.
However, there exist noise models such that $m\gamma$ is not negligible~\cite{Proctor}.

The analysis of Ref.~\cite{Magesan2012a} also holds for the more general
gate-dependent noise model $\cL_G \cG \cR_G$, where $\cL$ and $\cR_G$ are
gate-independent and gate-dependent noise processes acting from the left and
right (that is, after and before a ideal gate) respectively, so that some
optimization of $\gamma$ is possible.

However, we now show that the higher-order terms cannot be neglected
when fitting data in experimental regimes. Suppose the average survival
probability over all sequences of some fixed length $m$ is measured to be
$f_m$. From the zeroth-order model with $p = 1-2r$,
\begin{align}
r = \frac{1}{2} - \frac{1}{2}\sqrt[m]{\frac{f_m - B}{A}},
\end{align}
ignoring any finite measurement statistics. However, by
\cref{eq:perturbation_bound}, the first-order perturbations add some
$\delta\in[-m\gamma,m\gamma]$ to $f_m$, so that the true value of $r$ is
\begin{align}
r = \frac{1}{2} - \frac{1}{2}\sqrt[m]{\frac{f_m -B - \delta}{A}} 
\approx \frac{1}{2} - \frac{1}{2}\sqrt[m]{\frac{f_m - B}{A}} +
\frac{\delta}{2m (f_m - B)},
\end{align}
where we have assumed $m\gg 1$ for the approximation, so that the first-order
perturbation terms add a systematic uncertainty of $\gamma/2(f_m-B)$ to the
estimate of $r$.

Therefore the systematic ucertainty will dominate the estimate of the infidelity
unless $\gamma \approx r(\cE)$. To illustrate how impractical this requirement
is, and thus the value of the tighter analysis in \cref{sec:GD_RB}, consider the
following hypothetical implementation of the single-qubit Clifford group (that
is, the 24 elements of SU(2) that permute the single-qubit Paulis $X$, $Y$ and
$Z$ up to an overall sign). Supoose that half the elements are implemented with
only depolarizing noise, that is, $\tilde{\cG} =\cG\cD_\nu$, and the other half
are implemented with an additional rotation around the $z$-axis of the Bloch
sphere by some angle $\theta\in(0,\tfrac{\pi}{2})$, that is, $\tilde{\cG} =
\cG\cD_\nu\cZ_\theta$ where
\begin{align}
Z_\theta = \left(\begin{array}{cc} 1 & 0 \\ 0 & e^{i\theta} 
\end{array}\right).
\end{align}
For this noise model, the average noise is 
$\cE = \cD_\nu(I_4 + \cZ_\theta)/2$ by \cref{eq:av_noise} and so
\begin{align}
\gamma &= \frac{1}{2}\lV \cD_\nu - \cD_\nu(I_4 + \cZ_\theta)/2\rV_{1\to 1}
+ \frac{1}{2}\lV \cD_\nu\cZ_\theta - \cD_\nu(I_4 + \cZ_\theta)/2\rV_{1\to 1}
\notag\\
&= \frac{\nu}{2}\lV \cZ_\theta - I_4\rV_{1\to 1} \notag\\
r(\cE) &= \frac{1}{6}\Bigl(3 - 2\nu - \nu\cos\theta\Bigr)
\end{align}
using $1-r(\cE) = f(\cE,\cI) = 1/2 + \sum_{\sigma=X,Y,Z} \tr[\sigma \cE(\sigma)]/12$~\cite{Bowdrey2002}.
The maximization in the induced trace norm will be achieved by any state in the $xy$ plane (e.g., $(|0\ra + |1\ra)/\sqrt{2}$), giving
\begin{align}
\gamma = \nu\lv\sin(\theta/2)\rv.
\end{align}
For the estimate based on the zeroth-order model to be valid to within a factor
of two, we require $\gamma\leq r(\cE)$, that is, $\lv \theta \rv\lesssim
1-\nu$. For example, the systematic uncertainty in the estimate of $r$ is
$9r$ (i.e., under-reporting the error rate $r$ by an order of magnitude) when
$\nu = 0.99$ and gate-dependent errors account for 10\% of the total infidelity
($\theta = 0.09$ in the above model) or when $\nu=0.999$ and the gate-dependent
errors account for just 1\% of the total infidelity ($\theta = 0.009$ in the
above model).

One might expect that the above analysis is pessimistic and that the
contributions from gate-dependent perturbations at different sequence lengths
would average out. However, Ref.~\cite{Proctor} presented an explicit model 
wherein the contributions from gate-dependent perturbations are systematic.

\subsection{Subsequent analyses of RB}

Ref.~\cite{Proctor} demonstrated that the RB decay rate could differ from the
average gate infidelity of the average noise defined in \cref{eq:av_noise}
by orders of magnitude. Ref.~\cite{Proctor} also presented a significantly
improved analysis of the original RB protocol, which can be viewed as an 
application of the method of \cite{Chasseur2015} to the original RB protocol.
However, this analys has three crucial drawbacks.

First, no interpretation of the decay parameter was provided, leading to
frequent discussions of ``the RB number'' $r_{RB}$ at conferences. In the
original version, Ref.~\cite{Proctor} states that ``it is not clear that $r$ (a
function of the RB decay parameter) should be called (the) ``fidelity'', except
in the special case of gate-independent error maps,'' although this statement
has since been revised based on the analysis in the present paper.

Second, the analysis of Ref.~\cite{Proctor} is vulnerable to the criticism it
makes of Ref.~\cite{Magesan2012a}, namely, that the infidelity $r(\cE)$ of a
fixed decomposition of the noise (defined in \cref{eq:av_noise}) can differ from
$r_{RB}$ by orders of magnitude. This is problematic for the analysis of
Ref.~\cite{Proctor} (and also for that of Ref.~\cite{Chasseur2015}) because they
bound a perturbation from an exponential of the form $(1-r_{RB})^m$ (ignoring
small dimensional factors) by the diamond norm with respect to a fixed
decomposition of the noise, giving a perturbation on the order of
$\sqrt{r(\cE)}\gg r(\cE) \gg r_{RB}$. While some gauge optimization of their
bound is possible (especially using a gauge motivated by the present analysis),
it is unclear whether such optimization can close the gap between $r(\cE)$ and
$r_{RB}$. Fortunately, the present paper obviates the potential problem by
providing a much sharper bound.

Finally, even if gauge optimization allows the perturbation in their method to
be bounded by $\sqrt{r_{RB}}$, this still renders their analysis useless in two
key regimes, namely, for ``large'' error rates or in the linear part of the
decay.

For ``large'' error rates where $r_{RB} \approx 0.01$, which accounts for most
RB experiments to date, the best bound that can be put on the perturbation
without knowing the full experimental noise process is approximately
$\sqrt{r_{RB}} = 0.1$. With such a large perturbation, any sequence lengths
that have an average survival probability over $0.9$ or under $0.6$ are
effectively useless assuming SPAM values of $A=B=0.5$ (and the range of
values that contain any information decreases with $A$).

In the linear part of the decay, such as in the numerics of Ref.~\cite{Proctor},
the systematic uncertainty due to the perturbation dominates the decay.
For concreteness, consider the linear regime for a single qubit, 
\begin{align}
\bbE_{\vec{G}\in\bbG^m}(Q_{\vec{G}})\approx A + B -2A m [1-p(\cR\cL)] .
\end{align}
where $p(\cR\cL)$ is the true RB decay parameter from \cref{thm:GD_exact}.
If $p(\cR\cL)$ is estimated via the slope between $m=m_1$ and $m=m_2$, then the
unknown perturbation contributes a systematic uncertainty of approximately
$\sqrt{r(\cE)}/(m_1-m_2)$ to the estimate of $p$, so that essentially no
information is gained unless $m_1 - m_2 \gg \sqrt{r(\cE)}$. For the current
record value of $1-p \approx 10^{-6}$, this translates to requiring $m_2 - m_1
> 1000$ assuming $1-p(\cR\cL) \approx r(\cE)$. However, the assumption that
$1-p(\cR\cL) \approx r(\cE)$ is unjustified as shown by their analysis, so that
the required spacing of sequence lengths in order to overcome the systematic
uncertainty due to their bound is unknown, rendering their analysis impractical.
Consequently, the analysis of Ref.~\cite{Proctor} does not explain the empirical
success of RB to date or justify fitting to a single exponential without
further analysis.

\section{Conclusion}

We have proven that randomized benchmarking experiments result in a single exponential with a negligible perturbation under gate-dependent noise. 
Moreover, we proved that the decay parameter is the fidelity of the average noise between idealized gates.
More specifically, we can write $\tilde{\cG} = \cL\cG\cR_G$ where $\cL$ is a solution to \cref{eq:conditions}.
Then for trace-preserving noise the RB decay parameter is $\bbE_{G\in\bbG}[p(\cR_G\cL)]$, that is, the average survival probability between ideal gates.
As $\tilde{\cG}\to\cS^{-1}\tilde{\cG}\cS$ maps $\cL\to \cS\cL$ 
and $\cR\to\cR\cS^{-1}$ in \cref{thm:avnoise}, the decay parameter $p = \bbE_{G\in\bbG}(\cR_G\cL)$ is manifestly gauge invariant, answering the recent criticism of Ref.~\cite{Proctor}. 

By rigorously reducing the potential impact of gate-dependent errors, the present work enables randomized benchmarking to more reliably diagnose the presence of non-Markovian noise. 
More precisely, there are two key assumptions that may not hold well enough for the standard exponential decay model to be valid, namely, Markovianity and time-independence. 
However, the only effect that time-dependent noise can have is to alter the rate at which the survival probability decays~\cite{Emerson2005,Wallman2014}. 
Consequently, any ``revivals'' in RB experiments, that is, statistically significant increases in the average survival probability as $m$ increases~\cite{Epstein2014}, can be attributed directly to the noise being non-Markovian.

An intriguing observation from the present analysis is that gate-dependent
fluctuations may produce a significant deviation from the exponential fit for 
short sequences. This deviation is possible because the first few values 
of $m$ are the most sensitive to gate-dependent fluctuations, which may be as 
large as $0.1$ for $m=1$ and $0.01$ for $m=2$. Consequently, these sequence 
lengths should be omitted when fitting to \cref{eq:fidelity_decay}, but can also indicate the presence of gate-dependent noise.

Alternative protocols based on randomized benchmarking have been developed that 
also assume gate-independent noise~\cite{Emerson2007,Magesan2012,Wallman2015,
Wallman2015b,Wallman2016,Carignan-Dugas2015,Cross2016}. 
We leave applying the current analysis to these protocols to future work.

One main open problem of the current work is that the solutions $\cL$ and $\cR$ from \cref{thm:avnoise} are not guaranteed to be CPTP maps, even under rescaling.

\acknowledgments

The author acknowledges helpful discussions with Robin Blume-Kohout, Arnaud 
Carignan-Dugas, Joseph Emerson, Steve Flammia, Chris Granade, Richard Kueng, Robin Harper, Jonas Helsen, Thomas Monz, Jiaan Qi and Philipp Schindler.
This research was supported by the U.S. Army
Research Office through grant W911NF-14-1-0103. 
This research was undertaken thanks in part to funding from the Canada First
Research Excellence Fund.

\end{document}